\pdfoutput=1
\documentclass[11pt]{article}
\usepackage{fullpage}
\usepackage{amssymb}
\usepackage{amsmath}
\usepackage{amsthm}
\usepackage{graphicx}
\usepackage{color}

\newtheorem{theorem}{Theorem}[section]
\newtheorem{lemma}[theorem]{Lemma}

\newtheorem{corollary}[theorem]{Corollary}

\newtheorem{definition}[theorem]{Definition}

\def\compactify{\itemsep=0pt \topsep=0pt \partopsep=0pt \parsep=0pt}
\makeatletter
\renewcommand{\paragraph}{%
  \@startsection{paragraph}{4}%
  {\z@}{3.25ex \@plus 1ex \@minus .2ex}{-1em}%
  {\normalfont\normalsize\bfseries}%
}
\makeatother

\DeclareMathOperator{\rank}{rank}

\newcommand{\aset}[1]{\{#1\}}
\DeclareMathOperator{\mincut}{mincut}
\def\tcvector{{\Phi}}

\newcommand{\R}{\mathbb R}

\newcommand{\B}{\{ 0,1 \}}

\newcommand{\veps}{\varepsilon}

\title{Mimicking Networks and Succinct Representations \\ of Terminal Cuts%
\thanks{This work was supported in part by The Israel Science Foundation
(grant \#452/08), by a US-Israel BSF grant \#2010418,
and by the Citi Foundation.}
}

\author{Robert Krauthgamer
\qquad\qquad Inbal Rika
\\ Weizmann Institute of Science
\\ \texttt{\{robert.krauthgamer,inbal.rika\}@weizmann.ac.il}
}

\begin{document}

\maketitle

\thispagestyle{empty}
\setcounter{page}{0}

\begin{abstract}
Given a large edge-weighted network $G$ with $k$ terminal vertices,
we wish to compress it and store, using little memory,
the value of the minimum cut (or equivalently, maximum flow)
between every bipartition of terminals.
One appealing methodology to implement a compression of $G$
is to construct a \emph{mimicking network}:
a small network $G'$ with the same $k$ terminals,
in which the minimum cut value between every bipartition of terminals
is the same as in $G$.
This notion was introduced by Hagerup, Katajainen, Nishimura, and Ragde [JCSS '98],
who proved that such $G'$ of size at most $2^{2^k}$ always exists.
Obviously, by having access to the smaller network $G'$,
certain computations involving cuts can be carried out much more efficiently.

We provide several new bounds, which together narrow
the previously known gap from doubly-exponential to only singly-exponential,
both for planar and for general graphs.
Our first and main result is that every $k$-terminal planar network
admits a mimicking network $G'$ of size $O(k^2 2^{2k})$,
which is moreover a minor of $G$.
On the other hand, some planar networks $G$ require $|E(G')| \ge \Omega(k^2)$.
For general networks, we show that certain bipartite graphs
only admit mimicking networks of size $|V(G')| \geq 2^{\Omega(k)}$,
and moreover, every data structure that stores the minimum cut value between
all bipartitions of the terminals must use $2^{\Omega(k)}$ machine words.
\end{abstract}

\newpage

\section{Introduction}
\label{sec:intro}

These days, more than ever, we deal with huge graphs such as social networks, communication networks, roadmaps and so forth. But even when our main interest is only in a small portion of the input graph $G$, we still need to process all or most of it in order to answer our query, since the runtime and memory requirements of many common graph algorithms depend on the input (graph) size. Therefore, a natural question is whether we can find a smaller graph  $G'$ that exactly (or approximately) preserves some property of the original graph such as distances, cuts and connectivity. This basic concept is known as a \emph{graph compression} and was first introduced by Feder and Motwani \cite{FM95}, although their definition was slightly different technically. They require that the compressed graph has fewer edges than the original graph, and that each graph can be quickly computed from the other one.
They have demonstrated how this paradigm leads to significantly improved
running time by implementing it for several graph problems.

Yet another significant advantage of the compressed graph $G'$ is that it requires far less memory then storing the original graph $G$, which could be critical for machines with limited resources such as smartphones, assuming that the preprocessing can be executed in advance on much more powerful machines. This paradigm becomes indispensable when computations on the compressed graph are to be preformed repeatedly (after a one-time preprocessing).

We focus on cuts and flows, which are of fundamental importance in computer science, engineering, and operations research, because of their frequent usage in many application areas.
Specifically, we study the compression of a large graph $G$
containing $k$ ``important'' vertices (called terminals),
into a smaller graph $G'$ containing the same terminals,
while maintaining the following condition:
the minimum cut between every bipartition of the terminals
has exactly the same value in $G$ and in $G'$.
The above cut condition can be also stated in terms of maximum flow,
because it effectively deals with the single-source single-sink case,
for which we have the max-flow min-cut theorem.
We now turn to define this problem more formally,
restricting our attention (throughout) to undirected graphs.

A \emph{network} $(G,c)$ is a graph $G$ with an edge-costs function $c: E(G) \to \R^+$. The \emph{size} of a network is its number of vertices of $G$.
The network is called a \emph{$k$-terminal network}
if the graph $G$ has $k$ distinguished vertices called \emph{terminals},
denoted $Q=\{q_1,\ldots,q_k\}\subseteq V(G)$.
In such a network, a cut $(W,V(G)\setminus W)$ is said to be \emph{$S$-separating} if it separates the terminals subset $S\subset Q$ from the remaining terminals $\bar{S}:=Q \setminus S$, i.e. if $W\cap Q$ is either $S$ or $\bar S$.
When clear from the context,
$(W,V(G)\setminus W)$ may refer not only to a bipartition of the vertices,
but also to its corresponding \emph{cutset} (set of edges crossing the cut). %
The cost of a cut $(W,V(G)\setminus W)$ is the sum of costs of all the edges in the cutset.
We let $\mincut_{G,c}(S,\bar{S})$ denote an $S$-separating cut in the network $(G,c)$
of minimum cost (breaking ties arbitrarily).
With a slight abuse of notation,
we use the same notation to denote also the cost of the that cut.
We also omit the subscript $c$ when clear from the context.

\begin{definition}[Mimicking Network \cite{HKNR98}]
\label{mimicking}
   Let $(G,c)$ be a $k$-terminal network.
   A \emph{mimicking network} of $(G,c)$ is a $k$-terminal network $(G',c')$ with the same set of
   terminals $Q$, such that for all $S\subset Q$,%
   \footnote{Throughout, we omit the trivial exclusion $S\neq \emptyset,Q$.}
$$\mincut_{G',c'}(S,\bar{S}) = \mincut_{G,c}(S,\bar{S}).$$
\end{definition}

The above definition (albeit for directed networks) was introduced by
Hagerup, Katajainen, Nishimura, and Ragde \cite{HKNR98},
who proved that every $k$-terminal network $(G,c)$
admits a mimicking network of size at most $2^{2^k}$.
Subsequently, Chaudhuri, Subrahmanyam, Wagner, and Zaroliagis \cite{CSWZ00}
studied specific graph families, showing an improved upper bound $O(k)$
for graphs $G$ that have bounded treewidth.
For the special case of outerplanar graphs $G$,
the mimicking network $G'$ they construct is furthermore outerplanar.
Some of these previous results hold also for directed networks.

The only lower bound we are aware of on the size of mimicking networks is
$k+1$ for every $k>3$, even for a star graph, due to \cite{CSWZ00}.
For $k=4,5$ they further show a matching upper bound.
These results are summarized in Table \ref{tab:results}.
We mention that several other variants of the problem were studied in the
literature, in particular when cut values are preserved approximately,
see Section \ref{sec:related} for details.

\subsection{Our Results}
\label{sec:results}

\paragraph{(a) Upper bounds.}
We first prove (in Section~\ref{sec:UB}) a new upper bound for planar graphs,
which significantly improves over the bound that follows from previous work
(namely, $2^{2^k}$ known for general graphs \cite{HKNR98}).
See also Table~\ref{tab:results} for the known bounds.

\begin{theorem}\label{UB planar graphs}
Every planar $k$-terminal network $(G,c)$ admits a mimicking network of size at most $O(k^2 2^{2k})$, which is furthermore a minor of $G$.
\end{theorem}

Notice that our theorem constructs for an input graph $G$
a mimicking network that is actually a minor of it,
and thus preserves additional properties of $G$ such as planarity.

\begin{table}[t]
\vspace{.1in}
\begin{center}
\begin{tabular}{|l|cl|cl|}
\hline
Graph family
  & \multicolumn{2}{l|}{\ Lower bounds}
  & \multicolumn{2}{l|}{\ Upper bounds} \\
\hline
\hline
General graphs
  & $2^{\Omega(k)}$ & Theorem~\ref{LB general graphs}
  & $2^{2^k}$ & \cite{HKNR98} \\
\hline
Planar graphs
  & $|E(G')|\ge \Omega(k^2)$ & Theorem~\ref{LB planar graph}
  & $O(k^2 2^{2k})$ & Theorem~\ref{UB planar graphs} \\
\hline
Bounded treewidth
  & &
  & $O(k)$ & \cite{CSWZ00} \\
\hline
Star graphs
  & $k+1$ & \cite{CSWZ00}
  & & \\
\hline
\end{tabular}
\end{center}
  \caption{Known bounds for the size of mimicking networks}
  \label{tab:results}
\vspace{-.1in}
\end{table}

\paragraph{(b) Lower bounds.}
We further provide (in Section~\ref{sec:LB}) two nontrivial lower bounds.
See Table~\ref{tab:results} for comparison with the known bounds.
The following theorem addresses general graphs,
and narrows the previous doubly-exponential gap
(between $k+1$ and $2^{2^k}$) to be only singly-exponential.

\begin{theorem}\label{LB general graphs}
For every $k>5$ there exists a $k$-terminal network
such that every mimicking network of it has size $2^{\Omega(k)}$.
This holds even for bipartite networks with all the terminals
on one side and all the non-terminals on the other side.
\end{theorem}

The next theorem is for mimicking networks of planar graphs,
proving a lower bound on the number of \emph{edges}.
If the mimicking network is guaranteed to be sparse
(say planar, as is the case in our bound in Theorem \ref{UB planar graphs})
then we get a similar bound for the number of vertices.
But if the mimicking network could be arbitrary (e.g., a complete graph)
we do not know how to prove it cannot have $O(k)$ vertices.

\begin{theorem}\label{LB planar graph}
For every $k>5$ there exists a planar $k$-terminal network
such that every mimicking network of it has at least $\Omega(k^2)$ edges.
\end{theorem}

\emph{Remark.}
Very recently, we were informed of new results, obtained independently of ours, 
by Khan, Raghavendra, Tetali and V{\'e}gh \cite{KRTVdraft}.
Their results include improved upper bounds for general graphs 
(albeit still doubly-exponential in $k$),
for trees, and for bounded treewidth graphs,
as well as lower bounds that are comparable to ours.

\paragraph{(c) Succinct data structures.}
Our final result is an alternative formulation of graph compression
as the problem of storing succinctly (i.e., summarizing or sketching)
all the $2^k$ terminal cuts in a $k$-terminal network.

\begin{definition}\label{dfn:TCscheme}
A \emph{terminal-cuts (TC) scheme} is a data structure
that uses storage (memory) $M$ to support the following two operations on
a $k$-terminal network $(G,c)$, where $n=|V(G)|$ and $c:E(G) \to \aset{1,\ldots,n^{O(1)}}$.
\begin{enumerate} \compactify
  \item Preprocessing $P$, which gets as input the network and builds $M$.
  \item Query $Q$, which gets as input a subset of terminals $S$,
and uses $M$ (without access to $(G,c)$) to output $\mincut_{G,c}(S,\bar{S})$.
\end{enumerate}
\end{definition}
Observe that putting together the two conditions above gives
$Q(S;P(G)) = \mincut_{G,c}(S,\bar{S})$ for all $S\subset Q$.
The \emph{storage requirement} (or \emph{space complexity})
of the TC scheme is the (maximum) number of machine words used by $M$.
Since the value of every cut in $(G,c)$ is at most $n^{O(1)}$, and since we need to be able to represent every vertex in $G$, we shall count the size of the TC scheme in terms of machine words of $O(\log n)$ bits.
An obvious upper bound is $2^k$ machine words,
by explictly storing a list of all the cut values.
Perhaps surprisingly, we can show a matching lower bound for any data structure
using the technology developed to prove Theorem \ref{LB general graphs}.
We prove the following Theorem \ref{thm:LBdataStructure},
including an extension of it to randomized schemes, in Section \ref{app:LB4DS}.

\begin{theorem}\label{thm:LBdataStructure}
For every $k>5$, a terminal-cuts scheme for $k$-terminal networks
requires storage of $2^{\Omega(k)}$ machine words.
\end{theorem}

\subsection{Related Work}\label{sec:related}

Graph compression can be interpreted quite broadly,
and indeed it was studied extensively in the past,
with many results known for different graphical features
(the properties we wish to preserve).
For instance, in the context of preserving the graph distances,
concepts such as spanners \cite{PS89}
and probabilistic embedding into trees \cite{AKPW95,Bartal96},
have developed into a rich area with productive area,
and variations of it that involve a subset of terminal vertices
were studied more recently, see e.g. \cite{CE06,KZ12}.

In the context of preserving cuts (and flows), which is also our theme,
the problem of graph sparsification \cite{BK}
has recently seen an immense progress, see \cite{BSS08} and references therein.
Even closer to our own work are analogous questions that involve
a subset of terminals, and the goal is to find a small network that
preserves (the cost of) all minimum terminal cuts \emph{approximately}.
In particular, Chuzhoy \cite{Chuzhoy12} recently showed
a constant factor approximation using a network whose size depends
on (certain) edge-costs is in the original graph.
Another variation of our problem is that of a cut (and flow) sparsifier,
in which the compressed network should contain only $k$ vertices (the terminals)
and the goal is to minimize the approximation factor (called congestion),
see \cite{CLLM10,EGKRTT10,MM10} for the latest results.

\section{Upper Bound for Planar Graphs}
\label{sec:UB} %

In this section we prove Theorem \ref{UB planar graphs},
showing that every planar $k$-terminal network $(G,c)$ admits a mimicking network of size $O(k^2 2^{2k})$, which is in fact a minor of $G$.

\subsection{Technical Outline}\label{sec:UBtechniques}

Let $G$ be a planar $k$-terminal network, and assume it is connected.
Let $E_S=\mincut_{G,c}(S,\bar{S})$ be the cutset of a minimum $S$-separating cut in $(G,c)$, and let $\hat E$ be the union of $E_S$ over all subsets $S\subset Q$.
Removing the edges $\hat E$ from the graph $G$ disconnects it to
some number of connected components,
and we construct our mimicking network $G'$ by contracting
every such connected component into a single vertex.
It is easy to verify that these contractions maintain the minimum terminal cuts.
This method of constructing $G'$ resembles the one in \cite{HKNR98},
except that the sets of vertices that we unite are always connected,
hence our $G'$ is a minor of $G$.
We proceed to bound the number of connected components one gets in this way,
as this will clearly be the size of our mimicking network $G'$.

We first consider removing from $G$ a single cutset $\mincut_G(S,\bar S)$ (for arbitrary $S\subset Q$),
and show (in Lemma~\ref{bound number CCs in cutset})
that it can disconnect the graph into at most $k$ connected components.
We then extend this result to removing from $G$ two cutsets,
namely $\mincut_G(S,\bar S)$ and $\mincut_G(T, \bar T)$ (for arbitrary $S,T\subset Q$),
and show (in Lemma~\ref{bound number CCs in two cutsets})
such a removal can disconnect the graph into at most $3k$ connected components.
Next, we consider removing all the $m=2^{k-1}-1$ cutsets
of the minimum terminal cuts from $G$ (i.e., $\hat E$) from $G$.
However, naive counting of the number of resulting connected components,
which argues that every additional cutset splits each existing component
into at most $O(k)$ components,
would give us in total a poor bound of roughly $k^m$.

The crucial step here is to use the planarity of $G$
to improve the dependence on $m$ significantly,
and we indeed obtain a bound that is quadratic in $m$
by employing the dual graph of $G$ denoted $G^*$.
Loosely speaking, the cutsets in $G$ correspond to (multiple) cycles in $G^*$,
and thus we consider the dual edges of $\hat E$,
which may be viewed as a subgraph of $G^*$ comprising of (many) cycles.
We now use Euler's formula and the special structure of this subgraph of cycles;
more specifically, we count its vertices of degree $>2$, which turns out to
require the aforementioned bound of $3k$ for two sets of terminals $S,T$.
This gives us a bound on the number of faces in this subgraph
(in Lemma~\ref{bound number of faces in dual graph using meeting points}),
which in turn is exactly the number of connected component
in the primal graph (Corollary~\ref{cor:numCCs}).

\subsection{Preliminaries} \label{sec:UBprelims}

Recall that a graph is called a \emph{multi-graph} if we allow it to have parallel edges and loops.
A \emph{cycle} in a multi-graph $G$ is a sequence of edges $(u_0, v_0), \ldots, (u_{l-1}, v_{l-1})$ such that $v_i = u_{(i+1)\mod l}$ for all $i=0,\ldots,l-1$. The cycle is \emph{simple} if it contains $l$ distinct vertices and $l$ distinct edges. Note that two parallel edges define a simple cycle of length 2, and that a loop is a cycle of length 1 that contributes 2 to the degree of its vertex. A \emph{circuit} is a collection of cycles (not necessarily disjoint) $\mathcal{C}=\{C_1,\ldots,C_l\}$. Let $\mathcal{E(C)}=\bigcup_{i=1}^l C_i$ be the set of edges that participate in one or more cycles in the collection
(note it is not a multiset, so we discard multiplicities).
The cost of a circuit $\mathcal{C}$ is defined as $\sum_{e\in \mathcal{E(C)}}c(e)$.

For a graph $G$, let $CC(G)$ denote the set of connected components in the graph. In particular, if $CC(G)=\{P_1, \ldots, P_h\}$
then $V(G)=P_1\cup\cdots\cup P_h$ as a disjoint union.
For a subset of the vertices $W\subset V(G)$, let $\delta(W)$ denote the set of edges with exactly one endpoint in $W$, i.e. $\delta(W)=\{(u,v)\in E(G):\ u\in W, v\notin W\}$.
A vertex in $G$ with degree more than 2 will be called a \emph{meeting} vertex of $G$.
We introduce special notation for two (disjoint) sets of vertices:
\begin{align*}
  V_m(G) = \aset{v\in V:\ \deg(v)>2};
  \quad
  V_2(G) = \aset{v\in V:\ \deg(v)=2};
\end{align*}
and for two (disjoint) sets of edges:
\begin{align*}
  E_2(G) &:= \{(u,v)\in E(G):\ u,v\in V_2(G)\}; \\
  E_m(G) &:= \{(u,v)\in E(G):\ u\in V_m(G) \text{ or } v\in V_m(G) \}.
\end{align*}

\subsection{Proof of Theorem \ref{UB planar graphs}}
\label{sec:UBproof}

Let $(G,c)$ be a $k$-terminal network with terminals $Q=\{q_1, \ldots,q_k\}$, where $G$ is a connected plane graph with faces $F$ (if $G$ is not connected we can apply the proof for every connected component separately). We may assume, using small perturbation on the edges cost, that every two different subsets of edges in $G$ have different total cost. In the proof we will use the notations $E_S$ and $\hat E$ defined in Section \ref{sec:UBtechniques}.

\begin{lemma}[One cutset]
\label{bound number CCs in cutset}
For every subset of terminals $S$, the graph $G\setminus E_S$ has at most $k$ connected components.
\end{lemma}

\begin{proof}
If there are more than $k$ connected components then there is at least one connected component without any terminal vertex. Since $G$ connected, we can unite it to any other connected component by removing some edge from $E_S$. We get a new cutset that separates $S$ from $\bar{S}$ with smaller total cost than $E_S$ in contradiction to the minimality.
\end{proof}

\begin{lemma}[Two cutsets]
\label{bound number CCs in two cutsets}
For every two subsets of terminals $S$ and $T$, the graph $G\setminus (E_S\cup E_T)$ has at most $|CC(G\setminus E_S)|+|CC(G\setminus E_T)|+k$ connected components.
\end{lemma}

We illustrate this lemma in Figure~\ref{fig:replace W_T in W_S}. %
The idea is that if $G\setminus (E_S\cup E_T)$ has too many connected components,
then we can find one that contains no terminals,
and that moving it to the other side of (say) $G\setminus E_S$
contradicts the minimality of $E_S$.%

\begin{figure}[b]
\centering
\includegraphics[angle=0,width=0.5\textwidth]{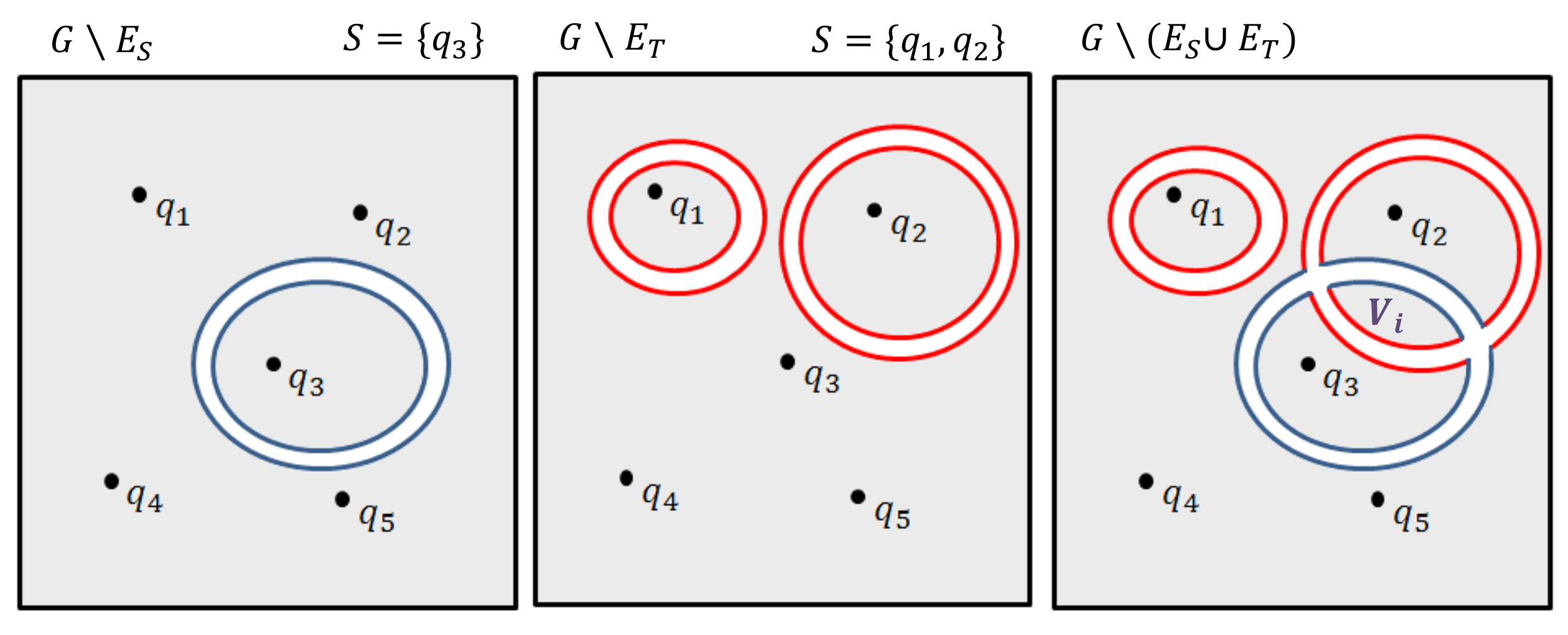}
\vspace{-.1in}
\caption{As depicted in gray, $G\setminus E_S$ has two connected components, $G\setminus E_T$ has three, and $G\setminus(E_S\cup E_T)$ has five. Notice the connected component $V_i$ of $G\setminus(E_S\cup E_T)$ contains no terminals.}
\label{fig:replace W_T in W_S}
\end{figure}

\begin{figure}[t]
\centering
\includegraphics[angle=0,width=0.3\textwidth]{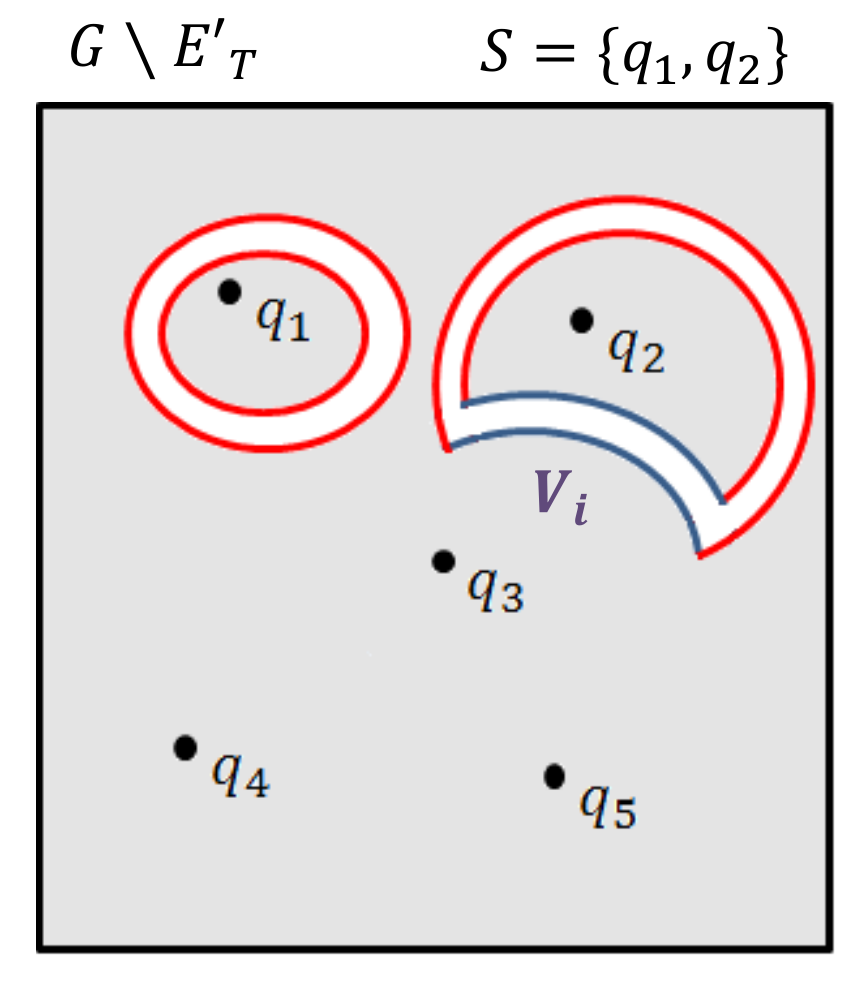}
\caption{$E'_T=(E_T\cup W_S)\setminus W_T$, where the red edges we removed are $W_T$, and the blue edges we add are $W_S$.}
\label{add empty CC to some other CC}
\end{figure}
\begin{proof}[Proof of Lemma~\ref{bound number CCs in two cutsets}]
Let $CC(G\setminus (E_S\cup E_T))=\{P_0,\ldots,P_h\}$.
For every $P_i$, we let $W_S(P_i):=\delta(P_i)\cap E_S$ be the set of edges
in $E_S$ that have exactly one of their endpoints in $P_i$,
and similarly $W_T(P_i):=\delta(P_i)\cap E_T$.
We can use the above notation to associate every connected components $P_i$
of $G\setminus (E_S\cup E_T)$, to one of the following four sets:
\begin{enumerate} \compactify
  \item $W_S(P_i)=\emptyset$; in particular, $P_i\in CC(G\setminus E_T)$.
  \item $W_T(P_i)=\emptyset$; in particular, $P_i\in CC(G\setminus E_S)$.
  \item $W_S(P_i)=W_T(P_i)$; in particular $P_i\in CC(G\setminus E_S)\cap CC(G\setminus E_T)$.
  \item $W_S(P_i)\neq \emptyset$, $W_T(P_i)\neq \emptyset$ and $W_S(P_i)\neq W_T(P_i)$; in particular $P_i\notin CC(G\setminus E_S)\cup CC(G\setminus E_T)$.
  \end{enumerate}
Every connected component that belongs to the last set (i.e. there are at least two different edges in $\delta(P_i)$, one from $E_T$ and one from $E_S$) will be called a \emph{mixed} connected component of $G\setminus (E_S\cup E_T)$.
Thus, the number of connected components in $G\setminus (E_S\cup E_T)$ is bounded by $|CC(G\setminus E_S)|+|CC(G\setminus E_S)|$ plus the number of mixed connected components of $G\setminus (E_S\cup E_T)$.

Assume towards contradiction that there are more than $k$ mixed connected components in $G\setminus (E_S\cup E_T)$. Therefore, there exists at least one mixed connected component, say with out loss of generality $P_0$, without any terminal in it. Since $P_0$ is a mixed connected component in $G\setminus(E_S \cup E_T)$ we know that $W_S(P_0)\neq \emptyset$, $W_T(P_0)\neq \emptyset$ and $W_S(P_0)\neq W_T(P_0)$. For simplicity from now on we will drop the $p_0$ and refer $W_S$ and $W_T$ to $W_S(P_0)$ and $W_T(P_0)$ correspondingly. By the perturbation on the edges cost the total cost of these two subsets must be different. Assume without loss of generality that $c(W_S) < c(W_T)$. We will replace the edges $W_T$ by the edges $W_S$ in the cutset of $T$ and call this new set of edges $E'_T$, i.e. $E'_T=(E_T\cup W_S)\setminus W_T$. It is clear that $c(E'_T) < c(E_T)$. We will prove that $E'_T$ is also a cutset that separate $T$ from $\bar{T}$ in the graph $G$, contradicting the definition of $E_T$. See Figures~\ref{fig:replace W_T in W_S} and~\ref{add empty CC to some other CC}.

Denote $CC(G\setminus E_T)=\{P'_0,\ldots,P'_{h'}\}$ and assume without loss of generality that the set of edges $W_S$ connects the connected component $P_0$ and the $t$ connected components $P_1,\ldots P_t$ of $G\setminus (E_S\cup E_T)$ into one connected component $P'_0$ in $G\setminus E_T$. Therefore, by adding the edges $E_S\setminus(W_S\cup W_T)$ to the graph $G\setminus (E_S\cup E_T)$ We will get the graph $G'=G\setminus (E_T\cup W_S)$ and its connected components will be $P_0,P_1,\ldots,P_t, P'_1,\ldots,P'_{h'}$. Since the graph $G'$ do not contains any edge from $E_T$, the sets $T$ and $\bar{T}$ are still separated.

Now it remain to add the edges $W_T$ to the graph $G'$ in order to get the desirable graph $G\setminus E'_T$. Assume without loss of generality that $P'_0$ contains terminals from $T$. Then, by the minimality of $E_T$, if edges from $W_T$ connect between $P'_0$ and $P'_i$, then the terminals of $P'_i$ are from $\bar{T}$. In particular, adding the edges $W_T$ to $G'$ will connect $P_0$ to some connected components $P'_i$ that contains only terminals from $\bar{T}$. Since $P_0$ does not contains any terminals, the connected component that was combined by the edges $W_T$ contains only terminals from $\bar{T}$, and so $E'_T$ separate between $T$ and $\bar{T}$.
\end{proof}

\paragraph{Planar duality.}

Recall that every planar graph $G$ has a dual graph $G^*$,
whose vertices correspond to the faces of $G$,
and whose faces correspond to the vertices of $G$,
i.e., $V(G^*)=\{v^*_f: f\in F(G)\}$ and $F(G^*)=\{f^*_v : v\in V(G)\}$.
Every edge $e=(v,u)\in E(G)$ with cost $c(e)$ that lies on the boundary of two faces $f_1,f_2\in F(G)$ has a dual edge $e^*=(v^*_{f_1},v^*_{f_2})\in E(G^*)$ with the same cost $c(e^*)=c(e)$ that lies on the boundary of the faces $f^*_v$ and $f^*_u$.
For every subset of edges $H\subset E(G)$, let $H^*:=\{e^*: e\in H\}$ denote the subset of the corresponding dual edges in $G^*$.

The following theorem describes the duality between two different kinds of edge sets -- minimum cuts and minimum circuits -- in a plane multi-graph.
It is a straightforward generalization of the case of $st$-cuts
(which are dual to cycles) to three or more terminals.
We are not aware of a reference for this precise statement,
although it is similar to \cite{HS85b,Rao87}.
See also Figure \ref{planar graph and its dual}
(in page \pageref{planar graph and its dual}) for illustration.

\begin{theorem}[Duality between cutsets and circuits]
\label{duality cuts and circuits}
Let $G$ be a connected plane multi-graph, let $G^*$ be its dual graph,
and fix a subset of the vertices $W\subseteq V(G)$.
Then, $H\subset E(G)$ is a cutset in $G$ that has minimum cost among those
separating $W$ from $V(G)\setminus W$
if and only if the dual set of edges $H^*\subseteq E(G^*)$
is actually $\mathcal E(\mathcal C)$ for a circuit $\mathcal C$ in $G^*$
that has minimum cost among those separating the corresponding faces
$\{f^*_v: v\in W\}$ from $\{f^*_v: v\in V(G) \setminus W\}$.
\end{theorem}

\begin{figure}[h]
\centering
\includegraphics[angle=0,width=0.7\textwidth]{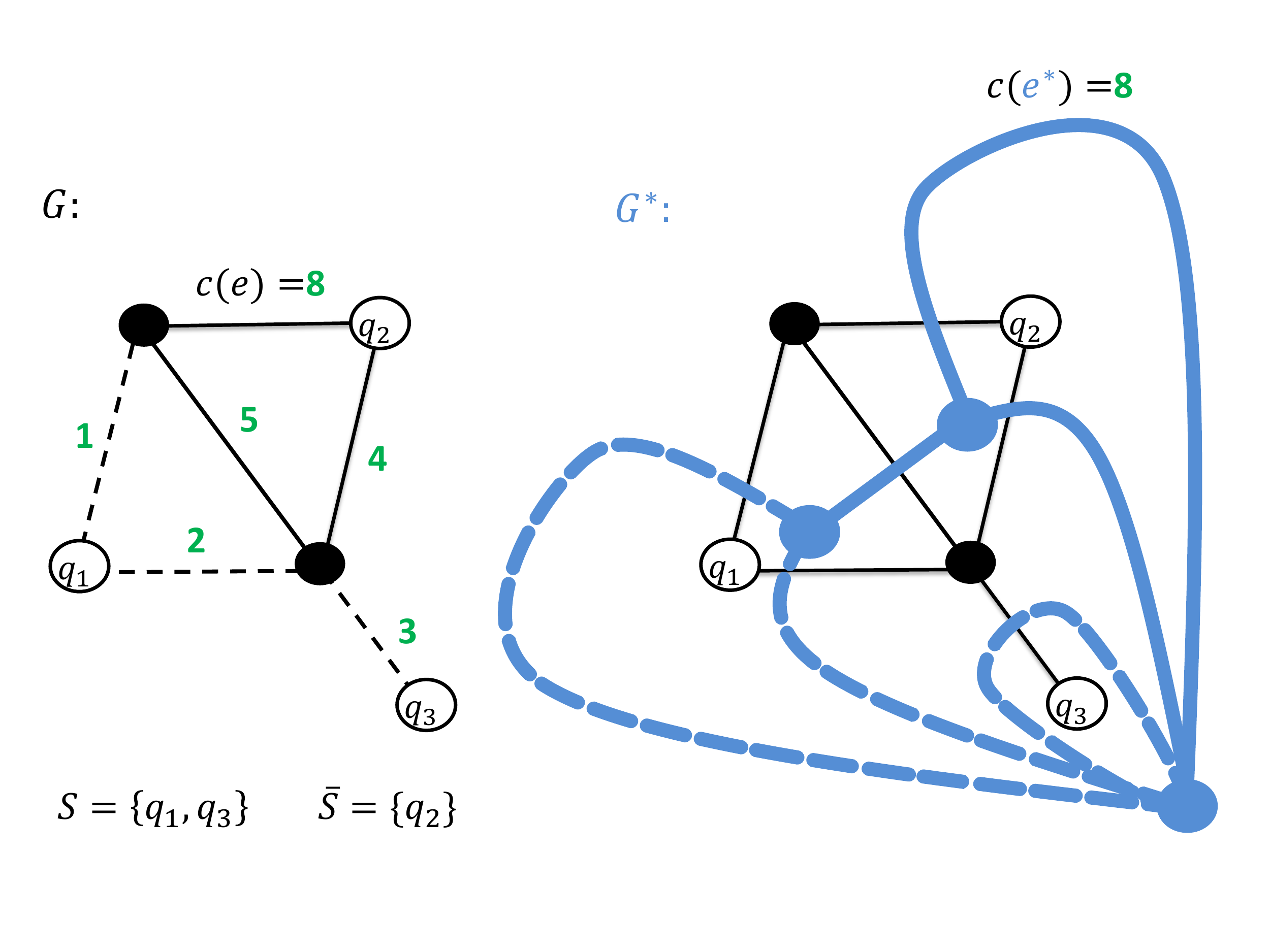}
\caption{A planar $3$-terminal network $G$ (in black),
with $E_S$ depicted as dashed edges.
The dual graph $G^*$ is shown in blue, with $E_S^*$ depicted as dashed edges.}
\label{planar graph and its dual}
\end{figure}

Recall that removing edges from a graph $G$ disconnects it into
(one or more) connected components.
The next lemma characterizes this behavior in terms
of the dual graph $G^*$.
Recall that $G[H]$ is a standard notation for the subgraph of $G$
induced by the subset (of edges or vertices) $H$.

\begin{lemma}[The dual of a connected component]
\label{connected components vs regions}
Let $G$ be a connected plane multi-graph, let $G^*$ be its dual,
and fix a subset of edges $H\subset E(G)$.
Then $P$ is a connected component in $G\setminus H$ if and only if
its dual set of faces $\{f^*_v: v\in P\}$ is a face of $G^*[H^*]$.
\end{lemma}

\paragraph{Leveraging the planarity.}
We proceed with the proof of Theorem \ref{UB planar graphs},
and now use the duality of planar graphs.
In the following corollary we will deal with the dual graph $G^*[E_S^*\cup E_T^*]$ for arbitrary two subsets of terminals $S$ and $T$.

\begin{corollary}\label{meeting vertices in two circuits}
For all $S,T\subset Q$, the graph $G^*[E_S^*\cup E_T^*]$ has at most $6k$ meeting vertices.
\end{corollary}

\begin{proof}[Proof of Corollary ~\ref{meeting vertices in two circuits}]
According Lemmas \ref{bound number CCs in cutset} and \ref{bound number CCs in two cutsets}, the graph $G\setminus (E_S\cup E_T)$ has at most $|CC(G\setminus E_S)|+|CC(G\setminus E_T)|+k \leq 3k$ connected components. By Lemma \ref{connected components vs regions} every connected component in $G\setminus(E_S \cup E_T)$ corresponds to a face in $G^*[E_S^*\cup E_T^*]$. Therefore, $G^*[E_S^*\cup E_T^*]$ has at most $3k$ faces.

By the duality of cuts and circuits, every set of edges $E_S^*$ is a circuit. Therefore, every vertex $v$ appearing in these edges $E_S^*$, has degree at least 2. $E_S^*\cup E_T^*$ is circuit as well, and all its vertices have degree at least 2, i.e. $V(G^*[E_S^*\cup E_T^*])=V_2(G^*[E_S^*\cup E_T^*])\cup V_m(G^*[E_S^*\cup E_T^*])$. To simplify the notation we denote $G^*_{ST}=G^*[E_S^*\cup E_T^*]$. By Handshaking lemma,
\begin{align*}
 2|E(G^*_{ST})|= \sum_{v\in V(G^*_{ST})}\deg(v) %
 \geq 3|V_m(G^*_{ST})|+2|V_2(G^*_{ST})| %
 = 2|V(G^*_{ST})| + |V_m(G^*_{ST})|.
\end{align*}
By Euler's formula
\begin{align*}
3k \geq |F(G^*_{ST})|=|E(G^*_{ST})|-|V(G^*_{ST})|+|CC(G^*_{ST})|+1 \geq %
  \frac{1}{2} |V_m(G^*_{ST})|,
\end{align*}
and the corollary follows.
\end{proof}

Recall that in Section \ref{sec:UBtechniques} we defined
$\hat{E}:=\bigcup_{S\subset Q}E_S$,
and denote its set of dual edges by $\hat{E}^*:=\{e^*:\  e\in \hat{E}\}=\bigcup_{S\subset Q}E_S^*$.

\begin{lemma}\label{bound number of faces in dual graph using meeting points}
The graph $G^*[\hat{E}^*]$ has at most $O(k^2 2^{2k})$ faces.
\end{lemma}
\begin{proof}[Proof of Lemma~\ref{bound number of faces in dual graph using meeting points}]
Using Theorem \ref{duality cuts and circuits} we get that for every $S\subset Q$, $E_S$ is a minimum cutset in $G$ if and only if $E^*_S$ (the dual set of edges of $E_S$) is a minimum circuit in $G^*$. Moreover, as defined in Section \ref{sec:UBproof} $\hat E ^*=\bigcup_{S\subset Q}E_S^*$. Thus, $\hat{E}^*$ is also a circuit, and so
\begin{align}
|V(G^*[\hat{E}^*])| = |V_2(G^*[\hat{E}^*])| + |V_m(G^*[\hat{E}^*])|, \label{eq1} \\
|E(G^*[\hat{E}^*])| = |E_2(G^*[\hat{E}^*])| + |E_m(G^*[\hat{E}^*])|. \label{eq2}
\end{align}
According the definitions and the Handshaking lemma we get that

\begin{equation}\label{inequality between E_2 to V_2}
|E_2(G^*[\hat{E}^*])|\leq|V_2(G^*[\hat{E}^*)|.
\end{equation}
By a union bound, the two following inequalities holds
\begin{align}
\textstyle
|V_m(G^*[\hat{E}^*])| \leq \sum_{S\subset Q}|V(G^*[E_S^*]) \cap V_m(G^*[\hat{E}^*])| \label{eq3},\\
\textstyle
|E_m(G^*[\hat{E}^*])| \leq \sum_{S\subset Q}|E(G^*[E_S^*]) \cap E_m(G^*[\hat{E}^*])|. \label{eq4}
\end{align}

Fix a subset $S$. We will start by bounding the set of vertices $V(G^*[E_S^*])\cap V_m(G^*[\hat{E}^*])$. Fore every vertex $v$ in $V(G^*[E_S^*])\cap V_m(G^*[\hat{E}^*])$ there exists a subset $T$ such that $v$ is also in $V(G^*[E_S^*])\cap V_m(G^*[E_S^*\cup E_T^*])$. According to Corollary \ref{meeting vertices in two circuits}, $|V_m(G^*[E_S^*\cup E_T^*])| \leq 6k$. Therefore $|V(G^*[E_S^*])\cap V_m(G^*[E_S^*\cup E_T^*])|\leq 6k$, and by union bound on all the subsets $T$ we get $|V(G^*[E_S^*]) \cap V_m(G^*[\hat{E}^*])| \leq 6k2^k $.

We will now move to bound $E(G^*[E_S^*]) \cap E_m(G^*[\hat{E}^*])$. By Lemma \ref{bound number CCs in cutset} there are at most $k$ cycles that cover the graph $G^*[E_S^*]$, so every vertex in $V(G^*[E_S^*])\cap V_m(G^*[\hat{E}^*])$ can be shared by at most $k$ cycles of $G^*[E_S^*]$, which bound the degree of every vertex in $G^*[E_S^*]$ by $2k$. Thus

\begin{equation}
|E(G^*[E_S^*]) \cap E_m(G^*[\hat{E}^*])|\leq 2k |V(G^*[E_S^*]) \cap V_m(G^*[\hat{E}^*])| =O(k^2 2^k) \label{eq5}
\end{equation}

We can bound $|CC(G^*[\hat{E}^*])|$ by extending the argument in Lemma \ref{bound number CCs in cutset}. Assume toward contradiction that $|CC(G^*[\hat{E}^*])|\geq k+1$. Thus, there exists at least one connected component $P$ in $G^*[\hat{E}^*]$ that does not contains any terminal face of $G^*$. By the construction of $\hat{E}^*$, there exists a subset $S$ such that $P$ contains at least one cycle $C$ of the circuit $E^*_S$. Since $P$ does not contain any terminal face, we can remove some edge $e^*$ of the cycle $C$ from the circuit $E_S^*$ and get circuit with smaller cost that separates between $f^*_S$ and $f^*_{\bar{S}}$ in contradiction.

Now by Euler's formula,
\begin{align*}
|F(G^*[\hat{E}^*])| = & |E(G^*[\hat{E}^*])| -|V(G^*[\hat{E}^*])|+1+|CC(G^*[\hat{E}^*])|\\
\leq & |E_m(G^*[\hat{E}^*])|-|V_m(G^*[\hat{E}^*])|+1+k & \mbox{by Eqns.~\eqref{eq1},\eqref{eq2},\eqref{inequality between E_2 to V_2}}\\
\leq & \textstyle
\sum_{S\subset Q}O(k^2 2^k) +1+k= O(k^2 2^{2k}), &\mbox{by Eqns.~\eqref{eq4},\eqref{eq5}}
\end{align*}
and the lemma follows.
\end{proof}

\begin{corollary}\label{cor:numCCs}
There are at most $ O(k^2 2^{2k})$ connected components in the graph $G\setminus \hat{E}$.
\end{corollary}

This corollary follows from
Lemma~\ref{bound number of faces in dual graph using meeting points}
by applying Lemma \ref{connected components vs regions} with $H=\hat{E}$.
We now complete the proof of Theorem \ref{UB planar graphs}.
Merge the vertices in each connected component of $G\setminus \hat{E}$ into a single vertex (formally, contract all the internal edges in each connected component) and call this new multi-graph $M$.
Notice there is at most one terminal vertex in each connected component. So a vertex in $M$, which corresponds to a connected component (of $G\setminus\hat E$) that contains some terminal vertex $q$, will be identified with that terminal $q$.
To be concrete, the vertices and the terminals of $M$ are the sets
$$V(M):=\{v_i:\ P_i\in CC(G\setminus \hat{E})\}$$
$$Q(M):=\{q=v_i:\  P_i\in CC(G\setminus \hat{E}) \text{ and } q\in P_i\}$$

In addition, $(v_i,v_j)$ is an edge in $M$ if there exist two vertices $u_i,u_j\in E(G)$ such that $u_i\in P_i$, $u_j\in P_j$ and $(u_i,u_j)$ is an edge in $G$. The cost of every edge $(v_i,v_j)\in E(M)$ is
$$c'(v_i,v_j):=\sum_{u_i\in P_i, u_j\in P_j:\  (u_i,u_j)\in E(G)}c(u_i,u_j).$$

It is easy to verify that $M$ is a minor of $G$ with $ O(k^2 2^{2k})$ vertices that includes the same $k$ terminals $Q$. We now prove that $(M,c')$ is a mimicking network of $G$ using the same argument as in \cite{HKNR98}, but applied to the connected components.
Fix a subset of terminals $S$. Since we only contract edges, every cut that separates between $S$ and $\bar{S}$ in $M$ has a cut in $G$ that separates between $S$ and $\bar{S}$ with the same cost, thus $\mincut_{M,c'}(S,\bar{S}) \geq \mincut_{G,c}(S,\bar{S})$. In the other direction, notice that by the construction of $M$, all the vertices in each connected components of $G\setminus \hat{E}$ are on the same side of the minimum $S$-separating cut in $G$. Thus, there is a cut in $M$ that separates between $S$ and $\bar{S}$ and has cost $\mincut_{G}(S,\bar{S})$.
Combining these together, we get the equality $\mincut_{M,c'}(S,\bar{S}) = \mincut_{G,c}(S,\bar{S})$ for every $S$,
and Theorem \ref{UB planar graphs} follows.

\section{Lower Bounds}
\label{sec:LB}

In this section we prove Theorems \ref{LB general graphs} as well as Theorem \ref{LB planar graph}.

\subsection{Techniques and Proof Outline}\label{sec:LBtechniques}

All our lower bounds are proved using the same technique,
which basically counts the number of ``degrees of freedom'' needed
to express all the relevant cut values.
Formally, we develop a certain machinery based on linear algebra,
which relates the size of any mimicking network to the rank of some matrix.

The lower bound proofs start by describing a $k$-terminal network $(G,c)$
that seems minimal %
in the sense that it does not admit a smaller mimicking network.
The networks used in Theorems \ref{LB general graphs} and \ref{LB planar graph}
are different, see Section \ref{sec:LB} for details.
We then identify the minimum cost $S$-separating cuts for all (or some)
$S\subset Q$, and capture this information in a matrix.

\begin{definition}[Incidence matrix between cutsets and edges]
\label{def: cutset-edge matrix}
Let $(G,c)$ be a $k$-terminal network,
and fix an enumeration $S_1,\ldots,S_m$ of all $m=2^{k-1}-1$
distinct and nontrivial bipartitions $Q=S_i\cup \bar S_i$.
The \emph{cutset-edge incidence matrix} of $(G,c)$ is the matrix
$A_{G,c}\in \B^{m\times E(G)}$ given by
$$(A_{G,c})_{i,e}=
  \begin{cases}
    1 & \text{if $e\in \mincut_(G,c)(S_i,\bar S_i)$;} \\
    0 & \text{otherwise.}
  \end{cases}
$$
\end{definition}

We also define the vector of minimum-cut values between every
bipartition of terminals
$$
  {\tcvector_{G,c}}
  =\left( \begin{array}{c}
      \mincut_{G,c}(S_1,\bar{S_1}) \\
      \vdots \\
      \mincut_{G,c}(S_m,\bar{S_m}) \\
    \end{array}
    \right)
    \in \R^m.
$$
Here and throughout,
we shall omit the subscript $c$ when it is clear from the context.
Observe that if we think of the edge costs $c$ as a column vector $\vec{c}\in (\R^+)^{E(G)}$,
then $A_G\cdot \vec{c}={\tcvector_G}$.
For a given $S\subset Q$,
a minimum $S$-separating cut $(W,V(G)\setminus W)$ is called \emph{unique}
if all other $S$-separating cuts have a strictly larger cost.

The core of our analysis is the next lemma,
as it immediately provides a lower bound on the size of any mimicking networks;
the theorems would follow by calculating the rank of $A_G$.

\begin{lemma}[Main Technical Lemma]
\label{rank A is root size mimicking}
Let $(G,c)$ be a $k$-terminal network.
Let $A_G$ be its cutset-edge incidence matrix, and assume that
for all $S\subset Q$ the minimum $S$-separating cut of $G$ is unique.
Then there is for $G$ an edge-costs function $\hat{c}: E(G) \to \R^+$,
under which every mimicking network $(G',c')$ satisfies
$|E(G')| \geq \rank(A_{G,{c}})$.
\end{lemma}
Notice that the bound is proved not for $(G,c)$ but rather for $(G,\hat c)$;
indeed, the edge-costs $\hat c$ are a small random perturbation of $c$.
Thus, the proof of this lemma first shows that a small perturbation
does not change the cutset-edge incidence matrix, i.e. $A_{G,c}=A_{G,\hat c}$.
This is where the uniqueness property is used.
Next, fix a small graph $G'$ that can potentially be a mimicking network,
but without specifying its edge-costs $c'$;
now let $\mathcal E_{G'}$ be the event that $(G,\hat c)$
admits a mimicking network of the form $(G',c')$.
Since $G'$ has too few edges (whose costs are undetermined/free variables),
we can use linear algebra to show that $\Pr[\mathcal E_{G'}]=0$.
The lemma then follows by a union bound over the finitely many (unweighted)
graphs $G'$ of the appropriate size.

\subsection{Proof of Lemma \ref{rank A is root size mimicking}}%
\label{proof rank A is root size mimicking}

We turn to proving Lemma \ref{rank A is root size mimicking}.
Recall that this lemma considers
a $k$-terminal network $(G,c)$,
and assuming a certain (uniqueness) condition,
asserts that there is for $G$ a modified edge-costs function $\hat c$,
under which every mimicking network must have at least $\rank(A_{G,c})$ edges,
where $A_{G,c}$ is a cutset-edge incidence matrix of $(G,c)$.

The proof employs two lemmas and the following notation.
For $S\subset Q$, let $\Delta_{G,c}(S)\ge0$ be the difference between the two
smallest costs among all $S$-separating cuts in $G$.
Observe that if these two are not equal (i.e., $\Delta_{G,c}(S)> 0$)
then the minimum $S$-separating cut is said to be {unique} in $G$.
We also denote $\Delta_{G,c}:=\min_{S\subset Q}\Delta_{G,c}(S)$.

\begin{lemma}\label{min-cuts incidence matrix equality}
For every edge-costs function $w: E(G) \to [0,\frac{1}{\Delta_{G,c}|E(G)|}]$ the cutset-edge incidence matrix of $(G,c)$ is equal to the cutset-edge incidence matrix of $(G,c+w)$, i.e. $A_{G,c}=A_{G,c+w}$, where $c+w: e \to c(e)+w(e)$.
\end{lemma}

\begin{proof}
Let $w$ be an edge-costs function $w: E(G) \to [0,\frac{1}{\Delta_{G,c}|E(G)|}]$. Since $(G,c)$ and $(G,{c+w})$ have the same vertices and edges, every $S_i$-separating cut in $(G,c)$ is also a $S_i$-separating cut in $(G,{c+w})$ and vice versa. The value of every such cutset in $(G,{c+w})$ is ranged from the value of this cutset in $G$ to the value of this cutset in $G$ plus $\frac{1}{\Delta_{G,c}}$. In particular, $\mincut_{G,c}(S_i,\bar{S_i}) \leq \mincut_{G,{c+w}}(S_i,\bar{S_i}) \leq \mincut_{G,c}(S_i,\bar{S_i}) + \frac{1}{\Delta_{G,c}}$. Thus, $\mincut_{G,{c+w}}(S_i,\bar{S_i})$ is smaller (by at least $\frac{\Delta_{G,c}-1}{\Delta_{G,c}}$) than every cut that separates between $S_i$ and $\bar{S_i}$ in $G$. Therefore it must be the case that the cutsets of the minimum $S_i$-separating cuts in $(G,c)$ and in $(G,{c+w})$ are the same.
\end{proof}

We proceed with the proof of Lemma \ref{rank A is root size mimicking}.
Sample an edge-costs function $w: E(G) \to [0,\frac{1}{\Delta_{G,c}|E(G)|}]$
by independently choosing each $w(e)$ from that range uniformly at random.
By the above lemma, $A_{G,c}=A_{G,c+w}$ so in the rest of the proof we will omit the edge-costs function and denote this matrix by $A_G$. Now we argue that every mimicking network of $(G,c+w)$ must has at least $r:=\rank(A_G)$ edges.
Consider some network $G'$ with $|E(G')|< r$,
and let's see if it can potentially be a mimicking network of $(G,{c+w})$.
Notice that every edge-costs function $c': E(G') \to \R^+$ for this $G'$
yields a cutset-edge incidence matrix $A_{G',c'}$ of size $m\times(r-1)$ (if some graph has less than $r-1$ edges we can pad the irrelevant columns with zeros).
Since this matrix has only ones and zeros in its entries,
there are only $2^{m(r-1)}$ such matrices.
The next lemma proves that for every fixed matrix $A\in \B^{m\times(r-1)}$,
the probability that there exists a edge-costs function $c': E(G') \to \R^+$ such that $A \cdot \vec{c'}=A_G \cdot (\vec{c}+\vec{w})=\tcvector_G$ is zero.

\begin{lemma}\label{not exist}
Fix a matrix $A\in\B^{m\times (r-1)}$, and let $W_{A_{G}}$ and $W_{A}$ be the span of the columns of $A_G$ and $A$, respectively.
If each $w(e)$ is independently sampled uniformly at random from $[0,\frac{1}{\Delta_{G,c}|E(G)|}]$, then
$$ \Pr_w[A_G \cdot (\vec{c}+\vec{w}) \in W_{A}]=0.$$
\end{lemma}

\begin{proof}
Without loss of generality let the first $r$ columns of the matrix $A_G$, $\{\vec{a_1}, \ldots, \vec{a_r}\}$, be the basis for the space $W_{A_G}$. Since $\rank(A)<r=\rank(A_G)$ we get that $\dim(W_{A})<\dim(W_{A_G})$. Thus there must be some basis vector of $W_{A_G}$, say without loss of generality $\vec{a_1}$, that not in the subspace $W_{A}$ and denote by $c(e_1)+w(e_1)$ to be its corresponding cost.

We will calculate the number of vectors in $W_{A}$ that can be expressed as linear combination with the vector $\vec{a_1}$. Let $f(\alpha)=\alpha \vec{a_1}+\sum_{i=2}^r (c(e_i)+w(e_i)) \vec{a_i}$. If there are at least two such vectors, $f(\alpha)$ and $f(\alpha')$ (where $\alpha,\alpha' \neq 0$) in $W_{A}$, then $\vec{a_1}$ will be in $W_{A}$ because $W_{A}$ is a subspace. So there is at most one $\alpha$ such that $f(\alpha) \in W_{A}$.

Since each $w(e_i)$ is sampled independently from a uniform distribution over $[0,\frac{1}{\Delta_{G,c}|E(G)|}]$, the probability that $c(e_1)+w(e_1)=\alpha$ is 0. By independence of $w(e_i)$ for all $i\in[r]$ we can sample $w(e_1)$ last which complete Lemma \ref{not exist}.
\end{proof}

To complete the proof of Lemma \ref{rank A is root size mimicking}, we will calculate the probability that there exists a mimicking network $(G',c')$ for the network $(G,c+w)$, such that $|E(G')|<r$.

\vspace{-.15in}
\begin{align*}
\textstyle
\Pr_w[\exists&\text{ mimicking network } (G',{c'}) \text{ with } |E(G')|<r] \\
  &= \textstyle \Pr_w[\exists A_{G',c'} \in \B^{m\times (r-1)} \text{ such that } A_{G',c'}\cdot \vec{c'}=A_G \cdot (\vec{c}+\vec{w})] \\
  &\leq \textstyle \Pr_w[\exists A_{G',c'} \in \B^{m\times (r-1)} \text{ such that } A_G \cdot (\vec{c}+\vec{w})\in W_{A_{G',c'}}] \\
  &\leq \textstyle \sum_{A\in \B^{m\times(r-1)}}\Pr_w[A_G \cdot (\vec{c}+\vec{w}) \in W_{A}] = 0,
\end{align*}
where the first equality is by definition of mimicking network, the following inequality is because the condition is necessary (but not sufficient), the second inequality is by a union bound over all possible matrices, and the final equality is by Lemma \ref{not exist}.
Denoting $\hat{c}=c+w$, we see that every mimicking network $(G',{c'})$ for the network $(G,\hat c)$ has at least $\rank(A_G)$ edges.
Lemma \ref{rank A is root size mimicking} follows.

\subsection{Lower bound for general graphs}\label{Lower bound for general graphs}

We now prove Theorem \ref{LB general graphs} which asserts that for every $k$ there exists a $k$-terminal network that its mimicking network must have $2^{\Omega (k)}$ non-terminals. The proof constructs a bipartite $k$-terminal network, with all its terminals on one side and all its non-terminals on the other side. As we will show, the rank of its cutset-edge incidence matrix is at least $2^{\Omega(k)}$, and the corresponding cuts are unique, hence applying Lemma \ref{rank A is root size mimicking} to this matrix will complete the proof of Theorem \ref{LB general graphs}.

\begin{proof}[Proof of Theorem \ref{LB general graphs}]
Consider a complete bipartite graph $G=(Q,U,E)$, where one side of the graph consist of the $k$ terminals $Q=\{q_1,\ldots,q_k\}$, the other side of the graph consists of $l=\binom{k}{\frac{2}{3}k}$ non-terminals $U=\{u_{S_1},\ldots,u_{S_l}\}$, with $S_1,\ldots,S_l$ denoting the different subsets of terminals of size $\frac{2}{3}k$.
The costs of the edges of $G$ are as follows.
Every non-terminal $u_{S_i}$ is connected by edges of cost $1$ to every terminal in $S_i$, and by edges of cost $2+\varepsilon$ to every terminal in $\bar{S_i}= Q \setminus S_i$,
for sufficient small $\varepsilon>0$, in fact $\varepsilon=\frac{1}{k}$ suffices.
Let $c(u_{S_i},q_j)$ denote the cost of edge $(u_{S_i},q_j)$, and define $c(u_{S_i},S_j):=\sum_{q\in S_j}c(u_{S_i},q)$.

\begin{lemma}\label{min-cut in bipartite graphs}
The minimum $S_i$-separating cut is obtained uniquely by the cut $(W,V(G)\setminus W)$ where $W=\{u_{S_i}\}\cup \bar{S}_i$ and $V(G)\setminus W=\{u_{S_j}:\ j\neq i\}\cup S_i$.
\end{lemma}

\begin{proof}
First, notice that for every $i\in[l]$ the total cost of all edges incident to $u_{S_i}$ is
\begin{equation}\label{edges in u_S}
   c(u_{S_i},Q) = c(u_{S_i},S_i) + c(u_{S_i},\bar{S_i}) = \frac{2k}{3}\cdot 1 + \frac{k}{3}(2+\varepsilon) = \frac{4k}{3}+\frac {k \varepsilon}{3}
\end{equation}

Consider such a set $S_i$, and let us calculate the minimum $S_i$-separating cut. Since non-terminals are not connected to each other, the decision is done separately for every non-terminal $u_{S_j}$ by simply comparing the costs of the edges $(u_{S_j},S_i)$ versus $(u_{S_j},\bar{S_i})$.
The crucial observation is that for non-terminal $u_{S_i}$:
  $$c(u_{S_i},{S_i}) = |S_i|\cdot 1 = \frac{2k}{3} < (2+\varepsilon)|\bar{S_i}|= c(u_{S_i},\bar{S_i})$$
For a non-terminal $u_{S_j}$ where $i\neq j$,
  $$c(u_{S_j},S_i) = |S_j\cap S_i|\cdot 1 + |S_i \setminus S_j|\cdot(2+\varepsilon) = |S_i|\cdot 1 + |S_i \setminus S_j|\cdot(1+\varepsilon) > \frac{2k}{3}+1 > c(u_{S_j},\bar{S_i})$$
where the last inequality is by (\ref{edges in u_S}) and because we choose $\varepsilon$ such that $\frac{k\varepsilon}{3}<1$.
It follows that for every $S_i$ the minimum $S_i$-separating cut will be $\{u_{S_i}\} \cup \bar{S_i}$ on one side, and $\{u_{S_j} :\  j\neq i \} \cup S_i$ on the other side, and moreover it is the unique minimizer.
\end{proof}

\begin{lemma}\label{rank AG=m}
Let $A_G$ be a cutset-edge incidence matrix of $G$.%
Then $\rank (A_G)\geq l$.
\end{lemma}

\begin{proof}
By definition, $A_G$ is a matrix of size $m\times kl$. Since $\binom{k}{\frac{2}{3}k}=l\leq m=2^{k-1}-1$, we need to show that $l$ rows of $A_G$ are linearly independent. Assume without loss of generality that the first $l$ rows of $A_G$ corresponds to the $l$ subsets of terminals of size $\frac{2}{3}k$, such that row $t$ corresponds to subset $S_t$. We will prove that these $l$ first rows of $A_G$ are linearly independent, i.e. $\sum_{t=1}^l\alpha_{t} {A_G}_t=\bar{0} \iff  \alpha_1=\ldots = \alpha_l=0$. We will focus on a column $j$ in $A_G$ that corresponds to some edge $(u_{S_i},q)$ where $q\in S_i$. In order to know how the $j$-th column in $A_G$ looks like, we need to know in which minimum cuts the edge $(u_{S_i},q)$ participates, i.e. we go over all the rows of $A_G$ and in each row $t$ we will ask if the edge is in the cutset of the minimum $S_t$-separating cut or not (if there is 1 or 0 in $({A_G})_{t,j}$).

According to the construction of $G$, if $q\in S_i$, then the terminal $q$ and the non-terminal $u_{S_i}$ are in different sides of the minimum $S_i$-separating cut, and the edge $(u_{S_i},q)$ in that cutset. For some subset $S_t$, where $t \neq i$ and $q\in \bar{S_t}$, the side of the minimum cut that contains the terminal $q$ will be $\{u_{S_t}\} \cup \bar{S_t}$, and the other side that contains $u_{S_i}$ will be $\{u_{S_f} :\  f\in[l],f\neq t \} \cup S_t$. Then again, the edge $(u_{S_i},q)$ will be in that cutset. It remain to look on some subset $S_t$, where $t \neq i$ and $q \in S_t$. The cut will be the same as above, but now both of the vertices will be in one side of the cut, i.e. $q,u_{S_i} \in \{u_{S_f} :\ f\in[l], f\neq t \} \cup S_t$, so the edge $(u_{S_i},q)$ will not participate in this cutset. In conclusion, The edge $(u_{S_i},q)$ participate in the cutset of the minimum $S_i$-separating cut, and in all the cutsets of the minimum $S_t$-separating cut such that $S_t$ do not contains the terminal $q$. Hence we will get that the entry $j$ (the column of $A_G$ that corresponds to the edge $(u_{S_i},q)$) in the vector $ \sum_{t=1}^l\alpha_{t} {A_G}_t$ is:
\begin{equation}
(\sum_{t=1}^l\alpha_{t} {A_G}_t)_j=\sum_{t=1}^l\alpha_{t} ({A_G})_{t,j}=\alpha_{i} +\sum_{t\in[l]:\ q\notin S_t }\alpha_{t}=0 \label{lin ind matrix}
\end{equation}
Every two different subsets $S_i$ and $S_{i'}$, have at least $\frac{1}{3}k$ terminals in common. In particular there exist some terminal $q$ contained in both of them. Looking at the entries corresponding to $(u_{S_i},q)$ and $(u_{S_{i'}},q)$ in the vector $ \sum_{h=1}^l\alpha_{h} {A_G}_h$ we have

$$\alpha_{i} +\sum_{t\in[l]:\ q\notin S_t }\alpha_{t}=0$$
$$\alpha_{{i'}} +\sum_{t\in[l]:\ q\notin S_t }\alpha_{t}=0$$

Thus $\alpha_{i}=\alpha_{{i'}}$ for every $i,i' \in [l]$. So we get the equation $\binom {k-1}{\frac{2}{3}k-1}\alpha_1=0$ in every entry in the vector equation $ \sum_{t=1}^l\alpha_{t} {A_G}_t=0$, and Lemma \ref{rank AG=m} follows.
\end{proof}

We can now complete the proof of Theorem \ref{LB general graphs}. Applying Lemma \ref{rank A is root size mimicking} to our bipartite graph $G$ and its cutset-edge incidence matrix $A_G$, we get that every mimicking network $G'$ of $G$ has at least $l=2^{\Omega(k)}$ edges. It follows that $|V(G')| \geq \sqrt{|E(G')|} \geq 2^{\Omega(k)}$.
\end{proof}

\subsection{Lower bound for planar graphs}\label{Lower bound for planar graphs}

In this section we prove Theorem \ref{LB planar graph},
which shows a planar $k$-terminal network,
every mimicking network of which must have at least $k^2$ edges.
The proof constructs a grid of size $O(k^2)$ with $2k$ terminals, and applies Lemma \ref{rank A is root size mimicking} on graph's cutset-edge incidence matrix.

\begin{proof}[Proof of Theorem \ref{LB planar graph}]
Construct a planar $2k$-terminal network $G$ with $2k$ terminals $Q=\{v_1, \ldots, v_k,\linebreak[2] h_1, \ldots, h_k\}$ as follows. Consider a grid with $k$ columns and $k$ rows. Let $u_{i,j}$ be the non-terminal vertex at the $i$th column and $j$th row of the grid. To every vertex $u_{1,j}$, for $1 \leq j \leq k$, we attach a terminal vertex $v_j$ of degree one, and at every vertex $u_{i,1}$, for $1 \leq i \leq k$, we attach a terminal vertex $h_i$ of degree one. From now on, we will refer to $i$ and $j$ as indices between $1$ to $k$, including $1$, excluding $k$.

The costs associated with the edges of $G$ are as follows: every edge that connects between a terminal to a non-terminal costs $k^4$. The cost of all the edges between the vertices $u_{i,k}$ and $u_{i+1,k}$, and between the vertices $u_{k,j}$ and $u_{k,j+1}$, is $k^4$. All the remaining vertical edges will have cost $1$, i.e. all the edges between $u_{i,j}$ and $u_{i+1,j}$. All the remaining horizontal edges, i.e. every edge between $u_{i,j}$ and $u_{i,j+1}$, will cost $1-\varepsilon_{i,j}$, where $\varepsilon_{i,j}=\frac{j}{k^4}$. Notice that for every $k>2$ the sum of all the $\varepsilon_{i,j}$ in $G$ is

\begin{equation}\label{less than 1}
\sum_{i,j=1}^{k-1}\varepsilon_{ij} \leq \frac{1}{k^4}\sum_{i,j=1}^k 2k=\frac{2k^3}{k^4}<1
\end{equation}

Denote by $S_{i,j}$ the subset of the terminals $\{h_1, \ldots, h_i,v_1, \ldots, v_j\}$. We are interested in all the $(k-1)^2$ minimum $S_{i,j}$-separating cuts. See the grid $G$ in Figure \ref{grid LB planar}.
\begin{figure}[ht]
\centering
\includegraphics[angle=0,width=0.75\textwidth]{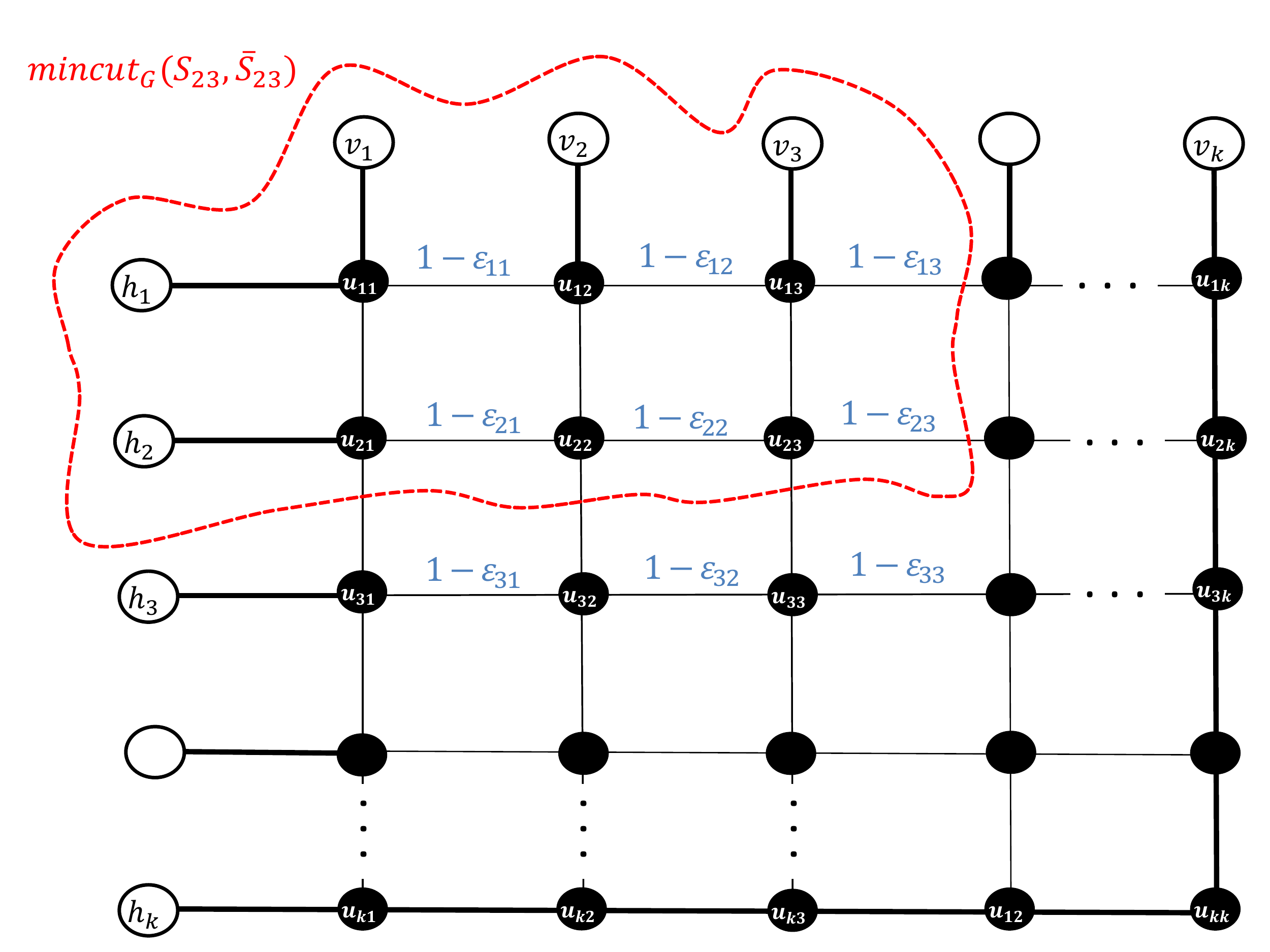}
\caption{The $2k$-terminal network, which used in Theorem \ref{LB planar graph}, with minimum $S_{23}$-separating cut (the red dashed line).
 All the vertical and horizontal bold edges has cost $k^4$, the remaining horizontal edges has cost $1-\veps_{i,j}$ and all the remaining vertical edges has cost 1.}
\label{grid LB planar}
\end{figure}

\begin{lemma}\label{min-cut}
The minimum $S_{i,j}$-separating cut is obtained uniquely by the cut $(W,V(G)\setminus W)$ where $W=S_{i,j}\cup \{u_{\alpha,\beta}: 1\leq \alpha \leq i, 1\leq \beta \leq j\}$.
\end{lemma}

\begin{proof}[Proof of Lemma~\ref{min-cut}]
Let $c_{i,j}$ be the cost of the $S_{i,j}$-separating cut $(W,V(G)\setminus W)$ described in the lemma. By a simple calculation, $c_{i,j}=i+j-\sum_{\alpha=1}^i \varepsilon_{\alpha,j}$. Assume towards contradiction that the above cut $(W, V(G)\setminus W)$ is not the minimum $S_{i,j}$-separating cut in $G$, i.e. $\mincut_G(S_{i,j},\bar{S}_{i,j}) < c_{i,j} < k $. Thus all the edges that are contained in $\mincut_G(S_{i,j},\bar{S}_{i,j})$ have costs less then $k$. In particular, the edges with cost $k^4$ are not contained in $\mincut_G(S_{i,j},\bar{S}_{i,j})$, so the two terminals $v_{k}$ and $h_{k}$ are connected (which means, not disconnected when we remove that cutset).

The cut $(W,V(G)\setminus W)$ contains $i$ horizontal edges and $j$ vertical edges. This is the minimal number of vertical and horizontal edges that need to be removed in the minimum cut in order to separate $S_{i,j}$ from $\bar{S}_{i,j}$. Otherwise, if we remove less then $i$ horizontal edges, there must be some terminal, $h_{\alpha}$, in $S_{i,j}$, such that no horizontal edges were removed from its row, thus $h_{\alpha}$ connected to the terminals $v_{k}$ and $h_{k}$ that in $\bar{S}_{i,j}$. The argument for $j$ vertical edges is similar.

Another observation is that the total cost of every $i+j+1$ or more edges in $G$ (with cost less then $k^4$) is not less than $i+j+1-\sum_{\alpha,\beta=1}^{k}\varepsilon_{\alpha,\beta}>i+j$, where the inequality is by Equation (\ref{less than 1}). We conclude that the minimum cut has exactly $j$ vertical edges and $i$ horizontal edges.

By now we know that the cutset $\mincut_G(S_{i,j},\bar{S}_{i,j})$ contains $i$ horizontal edges and $j$ vertical edges. Furthermore, we know that the cutset $(W,V(G)\setminus W)$ contains the first $i$ horizontal edges between the $j$th column to the $(j+1)$st column , and the first $j$ vertical edges between the $i$th row to the $(i+1)$st row. Thus, $\mincut_G(S_{i,j},\bar{S}_{i,j})$ must contains at least one different edge than the cut $(W,V(G)\setminus W)$. There are two cases:
\begin{enumerate}
  \item

  If $\mincut_G(S_{i,j},\bar{S}_{i,j})$ contains at least one vertical edge on some column $\beta>j$, then it contains no more than $j-1$ vertical edges from the columns between $1$ to $j$. As before, there exist some terminal that is connected to at least one terminal from $\bar{S}_{i,j}$. The same argument works for horizontal edge that removed from row $\alpha>i$. Hence, this case is impossible.

  \item

  If all the edges that participate in $\mincut_G(S_{i,j},\bar{S}_{i,j})$ are from the first $i$ rows and first $j$ columns. We will calculate the minimal value of a cut that we can obtain. As mentioned above, in order to separate we need to remove one edge from every column and from every row. The cost of all the vertical edges is identical so already need to pay $j$. Notice that in every row $\alpha$ the following inequality chain holds
  $$\varepsilon_{\alpha,1}<\varepsilon_{\alpha,2}<\ldots < \varepsilon_{\alpha,k-1}$$
  Therefore, the cost of the cheapest edge that we can take from that row is $1-\varepsilon_{\alpha,j}$. Summing all these costs we get $j+ \sum_{\alpha=1}^i(1-\varepsilon_{\alpha,j})$.

\end{enumerate}
From the second case we get that $\mincut_G(S_{i,j},\bar{S}_{i,j})=c_{i,j}$, and that the cut $(W, V(G)\setminus W)$ is the only cut with that value as we wanted.
\end{proof}

Proceeding with the proof of Theorem \ref{LB planar graph}, let $A_G$ be a cutset-edge incidence matrix of $G$ (see Definition \ref{def: cutset-edge matrix}).

\begin{lemma}\label{lemma:rank(AG) planar}
$\rank(A_G)\geq (k-1)^2$
\end{lemma}

\begin{proof}[Proof of Lemma~\ref{lemma:rank(AG) planar}]
Assume without loss of generality that the first $(k-1)^2$ columns of $A_G$ correspond to all the horizontal edges that their cost involve an $\varepsilon_{i,j}$ variable. We will order them according to their order in the grid from left to right, up to down. i.e. the first $(k-1)^2$ columns of $A_G$ will correspond to the edge costs in the following order:
$$1-\varepsilon_{1,1}\ , \ldots,\ 1-\varepsilon_{1,k-1}\ ,\ 1-\varepsilon_{2,1}\  ,\ldots,\ 1-\varepsilon_{2,k-1}\ ,\ldots,\ 1-\varepsilon_{k-1,1}\ ,\ldots,\ 1-\varepsilon_{k-1,k-1}$$
In addition, without loss of generality the first $(k-1)^2$ rows of $A_G$ correspond to the $(k-1)^2$ minimum $S_{i,j}$-separating cuts in $G$ which deals with the $(k-1)^2$ subsets of terminals we are interested in according to the following order:
$$S_{1,1}\ ,\ldots,\ S_{1,k-1}\ ,\ S_{2,1}\ ,\ldots,\ S_{2,k-1}\ ,\ldots,\ S_{k-1,1}\ ,\ldots,\ S_{k-1,k-1}  $$

We will show that the sub matrix of $A_G$ formed by first $(k-1)^2$ rows and columns of $A_G$ is a lower triangular matrix, which imply that the first $(k-1)^2$ columns are linearly independent. Given column $t$ that corresponds to $1-\varepsilon_{ij}$, we need to show that the entry $t,t$ is 1, and all the $t-1$ first entries are 0. As we set above, the $t$-th row of $A_G$ corresponds to the minimum $S_{i,j}$-separating cut. According to Lemma \ref{min-cut} the total costs of the horizontal edges that participate in the minimum $S_{i,j}$-separating cut is $\sum_{\alpha=1}^i(1-\varepsilon_{\alpha, j})$. Thus it is clear that entry $t,t$ is 1, because the edge $1-\varepsilon_{ij}$ participates in the minimum $S_{i,j}$-separating cut. It remains to show that all the $t-1$ first entries are 0. All the first $t-1$ rows correspond to subsets of terminals $S_{\alpha, \beta}$ such that $\alpha < i$ or $\alpha=i$ and $\beta<j$. As we saw above, the edge $1-\varepsilon_{i,j}$ participates only in all the minimum cuts of the subsets $S_{\alpha, j}$ where $\alpha \geq i$. Thus, there is 0 in all the first $t-1$ entries in the $t$-th column. So we prove that the first $(k-1)^2$ rows and columns of $A_G$ form a lower triangular matrix as we wanted, and the Lemma follows.
\end{proof}

To complete the proof of Theorem \ref{LB planar graph}, we apply Lemma \ref{rank A is root size mimicking} to our grid  $G$ and its cutset-edge incidence matrix $A_G$. We get that there exists an edge-costs function for $G$ such that every mimicking network of $G$ has at least $\rank(A_G)=\Omega(k^2)$ edges and the theorem follows.

\end{proof}

\section{Lower Bounds for Data Structures}
\label{app:LB4DS}

We can extend the definition of a (deterministic) TC scheme to a randomized one
by letting the two operations access a common source of random bits.
(We do not assume the random bits are stored explicitly in $M$,
even though it might be required in some implementations.)
We then change the requirement from the query operation to be
$$\Pr[Q(S;M)=\mincut_{G,c}(S,\bar{S})]\geq 2/3,$$
where the probability is taken over the data structure's random bits.
Our lower bound in Theorem \ref{thm:LBdataStructure} holds also for randomized
schemes, even those with shared randomness (that is not stored explicitly).

We now prove Theorem~\ref{thm:LBdataStructure}, which asserts that
a terminal-cuts scheme requires $2^{\Omega(k)}$ words in the worst-case.
Fix $k$ and let $(G,c)$ be the $k$-terminal bipartite graph constructed
in Section~\ref{Lower bound for general graphs}.
Recall that $l:=\binom{k}{2k/3}$ is the number of subsets of terminals
of size $2k/3$, each corresponding to a non-terminal in $G$.
The number of vertices in $G$ is $n:=k+l=2^{\Theta(k)}$,
and size of a machine word is $O(\log n)=\Theta(k)$ bits.
Assume towards contradiction there is a terminal-cuts scheme
that can handle every $k$-terminal network using less than ${l}/{100}$ bits.
For now, let us assume the scheme is deterministic.

Let $A_{G,c}$ be the cutset-edge incidence matrix of $(G,c)$.
By Lemma \ref{rank AG=m}, $\rank(A_{G,c})\geq l$.
Let us assume that the first $l$ columns of $A_{G,c}$ are linearly independent
(otherwise, we just reorder them),
and let $e_j$ denote the edge of $G$ corresponding to the $j$-th column of $A_{G,c}$.

Let $\mathcal{W}$ denote the collection of $2^{l}$ edge-costs functions
$w:E(G)\to \{0,\frac{1}{6k^2l}\}$ satisfying that $w(e_j)=0$ for all $j>l$.
As in Section~\ref{Lower bound for general graphs},
every function $w\in \mathcal{W}$ defines a graph $(G,{c+w})$,
whose cutset-edge incidence matrix is denoted $A_{G,{c+w}}$.
We can now apply Lemma \ref{min-cuts incidence matrix equality},
since $6k> \Delta_{G,c}$ and $|E(G)|=kl$,
and obtain that for all $w\in \mathcal{W}$ the network $(G,c+w)$ has the
same cutset-edge incidence matrix as $(G,c)$, i.e. $A_{G,c}=A_{G,{c+w}}$.
Using the above bound on the rank of $A_{G,c}$ we can deduce that
for every two different functions $w\neq w'\in \mathcal{W}$,
we have $A_{G,c} \cdot (\vec{c}+\vec{w})\neq A_{G,c} \cdot (\vec{c}+\vec{w'})$, i.e. there exists $S\subset Q$ such that $\mincut_{G,c+w}(S,\bar{S})\neq \mincut_{G,c+{w'}}(S,\bar{S})$.

Now, the assumed terminal-cuts scheme uses less than $l/100$ bits,
and thus, by the pigeonhole principle, there must be $w\neq w'\in \mathcal{W}$,
whose preprocessing results with the exact same memory image $M=P(G,c+w)=P(G,c+{w'})$.
Consequently, for all queries $S\subset Q$, the scheme will report the same
answer under inputs $(G,c+w)$ and $(G,c+w')$, which means that
$\mincut_{G,c+w}(S,\bar{S}) = \mincut_{G,c+w'}(S,\bar{S})$ and is a contradiction.

Notice that the edge costs of the graphs $(G,c+w)$ for $w\in\mathcal W$
can be easily scaled so that they are all in the range
$\aset{0,1,\ldots,n^{O(1)}}$.
We conclude that a terminals-cut scheme for $k$ terminals requires,
in the worst case, storage of at least
$\tfrac{l/100}{O(\log n)} \geq 2^{\Omega(k)}$ words.
This proves Theorem~\ref{thm:LBdataStructure} for deterministic schemes.

\paragraph{Proof for randomized schemes (sketch).}
The proof for randomized schemes follows the same outline,
the main difference being that we replace the simple collision argument
between $w\neq w'$, with well-known entropy (information) bounds.
First, the data structure's success probability can be amplified
to at least (say) $1-2^{2k}$, by straightforward independent repetitions,
while increasing the storage requirement by a factor of $O(k)$.
So assume henceforth this very high probability is the case.

Now let us choose $w\in\mathcal W$ at random,
which corresponds to choosing a random string of $l$ bits.
Using the data structure, one can retrieve with very high probability the value
$\mincut_{G,c+w}(S,\bar{S}) = A_{G,c} \cdot (\vec{c}+\vec{w})$.
Applying a union bound over all $2^k$ subsets $S\subset Q$,
with very high probability one would retrieves correctly all these values.
In this case, since the first $l$ columns of $A_{G,c}$ yield an invertible matrix,
we could actually recover the vector $w$ itself (with high probability).
But since $w$ is effectively a random string of $l$ bits,
it follows by standard entropy bounds that $M$ must have at least
$2^{\Omega(l)}$ bits,
and the theorem is completed just like for a deterministic scheme.

\section{Concluding Remarks}
\label{sec:conclusion}

Define
a \emph{generalized mimicking network} of a $k$-terminal network $(G,c)$
to be a $k$-terminal network $(G',c')$ with the same set of terminals $Q$,
such that for all disjoint $S,T\subset Q$, the minimum cost of a cut separating
$S$ from $T$ is the same, namely
$\mincut_{G',c'}(S,T) = \mincut_{G,c}(S,T).$
Although this definition increases the number of cuts that must be preserved,
our upper bound for planar graphs extends to this more general definition
(but with larger constants in the exponents),
and the same is true for the upper bound for general graphs by \cite{HKNR98}.

{%
\bibliographystyle{alphainit}
\bibliography{robi,drafts}

\newcommand{\etalchar}[1]{$^{#1}$}
\begin{thebibliography}{AKPW95}

\bibitem[AKPW95]{AKPW95}
N.~Alon, R.~M. Karp, D.~Peleg, and D.~West.
\newblock A graph-theoretic game and its application to the $k$-server problem.
\newblock {\em SIAM J. Comput.}, 24(1):78--100, February 1995.

\bibitem[Bar96]{Bartal96}
Y.~Bartal.
\newblock Probabilistic approximation of metric spaces and its algorithmic
  applications.
\newblock In {\em 37th Annual Symposium on Foundations of Computer Science},
  pages 184--193. IEEE, 1996.

\bibitem[BK96]{BK}
A.~A. Bencz{\'u}r and D.~R. Karger.
\newblock Approximating ${\rm s}$-${\rm t}$ minimum cuts in $\tilde {O}(n\sp
  2)$ time.
\newblock In {\em 28th Annual ACM Symposium on Theory of Computing}, pages
  47--55. ACM, 1996.

\bibitem[BSS09]{BSS08}
J.~D. Batson, D.~A. Spielman, and N.~Srivastava.
\newblock Twice-ramanujan sparsifiers.
\newblock In {\em 41st Annual ACM symposium on Theory of computing}, pages
  255--262. ACM, 2009.

\bibitem[CE06]{CE06}
D.~Coppersmith and M.~Elkin.
\newblock Sparse sourcewise and pairwise distance preservers.
\newblock {\em SIAM J. Discrete Math.}, 20:463--501, 2006.

\bibitem[Chu12]{Chuzhoy12}
J.~Chuzhoy.
\newblock On vertex sparsifiers with {S}teiner nodes.
\newblock In {\em 44th symposium on Theory of Computing}, pages 673--688. ACM,
  2012.

\bibitem[CLLM10]{CLLM10}
M.~Charikar, T.~Leighton, S.~Li, and A.~Moitra.
\newblock Vertex sparsifiers and abstract rounding algorithms.
\newblock In {\em 51st Annual Symposium on Foundations of Computer Science},
  pages 265--274. IEEE Computer Society, 2010.

\bibitem[CSWZ00]{CSWZ00}
S.~Chaudhuri, K.~V. Subrahmanyam, F.~Wagner, and C.~D. Zaroliagis.
\newblock Computing mimicking networks.
\newblock {\em Algorithmica}, 26:31--49, 2000.

\bibitem[EGK{\etalchar{+}}10]{EGKRTT10}
M.~Englert, A.~Gupta, R.~Krauthgamer, H.~R{\"a}cke, I.~Talgam-Cohen, and
  K.~Talwar.
\newblock Vertex sparsifiers: New results from old techniques.
\newblock In {\em 13th International Workshop on Approximation, Randomization,
  and Combinatorial Optimization}, volume 6302 of {\em Lecture Notes in
  Computer Science}, pages 152--165. Springer, 2010.

\bibitem[FM95]{FM95}
T.~Feder and R.~Motwani.
\newblock Clique partitions, graph compression and speeding-up algorithms.
\newblock {\em J. Comput. Syst. Sci.}, 51(2):261--272, 1995.

\bibitem[HKNR98]{HKNR98}
T.~Hagerup, J.~Katajainen, N.~Nishimura, and P.~Ragde.
\newblock Characterizing multiterminal flow networks and computing flows in
  networks of small treewidth.
\newblock {\em J. Comput. Syst. Sci.}, 57:366--375, 1998.

\bibitem[HS85]{HS85b}
D.~S. Hochbaum and D.~B. Shmoys.
\newblock An {$O(\vert V\vert ^2)$} algorithm for the planar {$3$}-cut problem.
\newblock {\em SIAM J. Algebraic Discrete Methods}, 6(4):707--712, 1985.

\bibitem[KRTV12]{KRTVdraft}
A.~Khan, P.~Raghavendra, P.~Tetali, and L.~A. V{\'e}gh.
\newblock On mimicking networks representing minimum terminal cuts.
\newblock Manuscript, July 2012.

\bibitem[KZ12]{KZ12}
R.~Krauthgamer and T.~Zondiner.
\newblock Preserving terminal distances using minors.
\newblock To appear in ICALP. Preliminary version available at
  \url{http://arxiv.org/abs/1202.5675}, 2012.

\bibitem[MM10]{MM10}
K.~Makarychev and Y.~Makarychev.
\newblock Metric extension operators, vertex sparsifiers and lipschitz
  extendability.
\newblock In {\em 51st Annual Symposium on Foundations of Computer Science},
  pages 255--264. IEEE, 2010.

\bibitem[PS89]{PS89}
D.~Peleg and A.~A. Sch{\"a}ffer.
\newblock Graph spanners.
\newblock {\em J. Graph Theory}, 13(1):99--116, 1989.

\bibitem[Rao87]{Rao87}
S.~Rao.
\newblock Finding near optimal separators in planar graphs.
\newblock In {\em 28th Annual Symposium on Foundations of Computer Science},
  pages 225--237. IEEE, 1987.

\end{thebibliography}
}

\end{document}